\documentclass[12pt]{article}
\usepackage[margin=25mm]{geometry}                
\usepackage{graphicx}
\usepackage{amssymb,amsmath,amsthm,comment}

\DeclareMathOperator{\vbl}{\mathrm{vbl}}
\DeclareMathOperator{\Tr}{\mathcal{T}}
\DeclareMathOperator{\E}{\mathsf{E}}
\DeclareMathOperator{\size}{\mathrm{size}}
\newtheorem{theorem}{Theorem}
\newtheorem{proposition}{Proposition}
\newtheorem{lemma}{Lemma}
\newtheorem{corollary}{Corollary}

\title{Probabilistic Constructions of Computable Objects and a Computable Version of Lov\'asz Local Lemma.}
\author{Andrei Rumyantsev, Alexander Shen\thanks{LIRMM CNRS \& UM~2, Montpellier; on leave from IITP RAS, Moscow. Supported by ANR NAFIT-008-01/2 grant. E-mail: \texttt{azrumuyan@gmail.com}, \texttt{alexander.shen@lirmm.fr}}}

\begin{document}
\maketitle

\begin{abstract}
A nonconstructive proof can be used to prove the existence of an object with some properties without providing an explicit example of such an object. A special case is a probabilistic proof  where we show that an object with required properties appears with some positive probability in some random process. Can we use such arguments to prove  the existence of a \emph{computable} infinite object? Sometimes yes: following~\cite{rumyantsev-positive}, we show how the notion of a layerwise computable mapping can be used to prove a computable version of Lov\'asz local lemma.
\end{abstract}

\section{Probabilistic generation of infinite sequences}

We want to show that one can use probabilistic arguments to prove the existence of inifinite objects (say, infinite sequences of zeros and ones) with some properties. So we should discuss first in which sense a probabilistic algorithm can generate such a sequence. The most natural approach: consider a Turing machine without input that has an (initially empty) work tape and a write-only output tape where machine prints bits sequentially. Using fair coin tosses, the machine generates a (finite or infinite) sequence of output bits. The output distribution of such a machine $T$ is some probability distribution $Q$ on the set of all finite and infinite sequences. Distributions on this set are determined by their values on cones $\Sigma_x$ (of all finite and infinite extensions of a binary string $x$). The function $q(x) = Q(\Sigma_x)$ that corresponds to the distribution $Q$ satisfies the following conditions:
    $$ q(\Lambda)=1;\quad q(x)\ge q(x0)+q(x1) \text{ \ for all strings $x$}.$$
(Here $\Lambda$ denotes the empty string.) Any non-negative real function that satisfies these conditions corresponds to some probability distribution on the set of finite and infinite binary sequences. The output distributions of probabilistic machines correspond to functions $q$ that are \emph{lower semicomputable}; this means that some algorithm, given a binary string $x$, computes a non-decreasing sequence of rational numbers that converges to $q(x)$.

Now we are ready to look at the classical result of de Leeuw -- Moore -- Shannon -- Shapiro: \emph{if some individual sequence $\omega$ appears with positive probability as an output of a probabilistic Turing machine, this sequence is computable}. Indeed, let this probability be some $\varepsilon>0$; take a rational threshold $r$ such that $r<\varepsilon<2r$, and consider some prefix $w$ of $\omega$ such that $q(w)<2r$. (Such a prefix exists, since $q(x)$ for prefixes of $\omega$ converges to $\varepsilon$.)
Starting from $w$, we can compute the next prefix of $\omega$ by finding a son where $q$ exceeds $r$. The correct son satisfies this condition, and no branching is possible: if for two sons the value exceeds $r$, then it would exceed $2r$ for the father.

This result can be interpreted as follows: \emph{if our task is to produce some specific infinite sequence of zeros and ones, randomization does not help} (at least if we ignore the computational resources). However, if our goal is to produce \emph{some} sequence with given properties, randomization can help. A trivial example: to produce a noncomputable sequence with probability $1$ it is enough to output the random bits.

All these observations are well known, see, e.g., the classical paper of Zvonkin and Levin~\cite{zvonkin-levin}. For a less trivial example, let us consider another result (proved by N.V.~Pet\-ri) mentioned in this paper: \emph{there exists a probabilistic machine that with positive probability generates a sequence $\omega$ such that \textup{(1)}~$\omega$ contains infinitely many ones; \textup{(2)}~the function $i\mapsto\text{the position of $i$-th one in $\omega$}$ has no computable upper bound}. (In the language of recursion theory, this machine generates with positive probability a characteristic sequence of a hyperimmune set.) A nice proof of this result was given by Peter Gacs; it is reproduced in the next section (we thank M.~Bondarko for some improvements in the presentation).

\section{Fireworks and hyperimmune sequences}

Consider the following puzzle. We come to a shop where fireworks are sold. After we buy one, we can test it in place (then we know whether it was good or not, but it is not usable anymore, so we have to buy a new one after that), or go home, taking the untested firework with us. We look for a probabilistic strategy that with 99\% probability wins in the following sense: it \emph{either finds a bad firework during the testing} (so we can sue the shop and forget about the fireworks) \emph{or takes home a good one}.

Here is the solution: \emph{take a random number $k$ in $0..99$ range, make $k$ tests \textup(less if the failure happens earlier\textup{);} if all $k$ tested fireworks were good, take home the next one}. To prove that it wins with $99\%$ probability, note that the seller does not get any information from our behavior: he sees only that we are buying and testing the fireworks; when we take the next one home instead of testing, it is too late for him to do anything. So his strategy is reduced to choosing some number $K$ of good fireworks sold before the bad one. He wins only if $K=k$, i.e., with probability at most~$1\%$.

Another description of the same strategy: we take the first firework home with probability $1/100$ and test it with probability $99/100$; in the second case, if the firework was good, we take the second one,  bringing it home with probability $1/99$ and testing it with probability $98/99$, etc.

Why this game is relevant? Assume we have a program of some computable function $f$ and want to construct probabilistically a total function $g$ not bounded by $f$ if $f$ is total; if $f$ is not total, any total function $g$ is OK for us. (It is convenient to consider a machine that constructs a total integer-valued function $g$ and then convert it into a bit sequence by putting $g(k)$ zeros after $k$th occurence of $1$ in the sequence.) We consider $f(0), f(1),\ldots$ as ``fireworks''; $f(i)$ is considered as a good one if the computation of $f(i)$ terminates. First we buy $f(0)$; with probability $1/100$ we ``take'' it and with probability $99/100$ we ``test'' it. \emph{Taking} $f(0)$ means that we run this computation until it terminates and then let $g(0):=f(0)+1$. If this happens, we may relax and let all the other values of $g$ be zeros. (If the computation does not terminate, i.e., if we have taken a bad firework, we are out of luck.) \emph{Testing} $f(0)$ means that we run this computation and at the same time let $g(0):=0$, $g(1):=0$, etc. until the computation terminates. If $f(0)$ is undefined, $g$ will be zero function, and this is OK since we do not care about non-total functions~$f$. But if $f(0)$ is defined, at some point testing stops, we have some initial fragment of zeros $g(0),g(1),\ldots,g(u)$, and then consider $f(u+1)$ as the next firework bought and test [take] it with probability $98/99$ [resp. $1/99$]. For example, if we decide to test it, we run the computation $f(u+1)$ until it terminates, and then let $g(u+1):=f(u+1)+1$. And so~on.

In this way we can beat a given computable function $f$ with probability arbitrarily close to~$1$. We need more: to construct with positive probability a function not bounded by \emph{any} total computable function. How can we do this? Consider all the functions as functions of two natural arguments, using a computable bijection between $\mathbb{N}$ and $\mathbb{N}^2$.  Use $i$th row  in the table of such a function to beat $i$th potential upper bound with probability $1-\varepsilon2^{-i}$. To beat the upper bound, it is enough to beat it in some row, so we can deal with all the rows in parallel, and get error probability at most $\varepsilon\sum_i 2^{-i}=\varepsilon$.

\section{Computable elements of closed sets}

Let us return now to the original question: can we use probabilistic arguments to construct a computable sequence with some property? As we have seen, if we are able to construct a probabilistic machine that generates some \emph{specific} sequence with positive probability, we can then conclude that this specific sequence is computable. However, we do not know arguments that follow this scheme, and it is difficult to imagine how one can describe a specific sequence that it is actually computable, and prove that it has positive probability --- without actually constructing an algorithm that computes it.

Here is another statement that may be easier to apply. In this statement we use the standard topology on the Cantor space of infinite $0$-$1$-sequences (=~the product topology on $\{0,1\}\times\{0,1\}\times\ldots$).

\begin{proposition}\label{closed-almost-everywhere}
Let $F$ be a closed set of infinite sequences. Let $M$ be a probabilistic machine \textup(without input\textup) whose output sequence belongs to $F$ with probability $1$. Then $F$ contains a computable element.
\end{proposition}

\begin{proof}
Indeed, consider the output distribution $Q$ of the machine $M$ and take a computable branch in the binary tree along which the probabilities $Q(\Sigma_x)$ remain positive (this is possible since the function $Q(\Sigma_x)$ is lower semicomputable). We get some computable sequence~$\omega$. If $\omega\notin F$, then some prefix of $\omega$ has no extensions in $F$ (recall that $F$ is closed). This prefix has positive probability by construction, so our machine cannot generate elements in $F$ with probability~$1$. This contradiction shows that $\omega\in F$.
\end{proof}

In the following sections we give a specific example when this approach (in a significantly modified form) can be used.

\section{Lov\'asz local lemma}

Let $\mathcal{X}=(x_1,x_2,\ldots)$ be a sequence of mutually independent random variables; each variable $x_i$ has a finite range $X_i$. (In the simplest case $x_i$ are independent random bits.) Consider some family $\mathcal A=(A_1,A_2,\ldots)$ of events; each $A_i\in\mathcal{A}$ depends on a finite set of variables, denoted $\vbl(A_i)$.   Informally speaking, Lov\'asz local lemma (LLL) guarantees that these events do not cover the entire probability space if each of them has small probability and the dependence between the events (caused by common variables) is limited. Intuitively, the events are undesirable for us and we want to avoid all of them; LLL says that this is possible (with positive probability).

To make the statement exact, we need to introduce some terminology and notation. Two events $A,B\in\mathcal{A}$ are \emph{disjoint} if they do not share variables, i.e., if $\vbl(A)\cap\vbl(B)=\varnothing$; otherwise they are called \emph{neighbors}. For every $A_i\in\mathcal{A}$ let $N(A_i)$ be the neighborhood of $A_i$, i.e., the set of all events $A_j\in \mathcal{A}$ that have common variables with $A_i$. Each event is its own neighbor.

\begin{proposition}[Finite version of LLL, \cite{original-lovasz}]
     \label{finite}
Consider a finite family $\mathcal{A}=(A_1,A_2,\ldots,A_N)$ of events. Assume that for each event $A_i\in\mathcal{A}$ a rational number $z_i\in (0,1)$ is fixed such that
    $$
\Pr [A_i]\le z_i\cdot\!\!\!\!\!\!\!\!\!\!\prod_{j\ne i,\, A_j\in N(A_i)}(1-z_j)
    $$
for all $i$. Then
$$\Pr[\text{none of the events $A_i$ happens}] \ge \prod_{i} (1-z_i)\eqno (*)$$
\end{proposition}

Note that we skip the event $A_i$ in the right hand side of the condition.

This bound was originally proved~\cite{original-lovasz} by a simple (though a bit misterious) induction argument.

Note that the product in the right hand side of $(*)$ is positive (though it can be very small), and therefore there exists an assignment that avoids all the forbidden events. This existence result can be easily extended to infinite families:

\begin{proposition}[Infinite version of LLL]
     \label{infinite}
Consider a sequence $\mathcal{A}=(A_1,A_2,\ldots)$ of events. Assume that each event $A_i\in\mathcal{A}$ has only finitely many neighbors in~$\mathcal{A}$, and that for each event $A_i\in\mathcal{A}$ a rational number $z_i\in (0,1)$ is fixed such that
    $$
\Pr [A_i]\le z_i\cdot\!\!\!\!\!\!\!\!\!\!\prod_{j\ne i,\, A_j\in N(A_i)}(1-z_j)
    $$
for all $i$. Then there exists an assignment that avoids all the events $A_i\in\mathcal{A}$.
\end{proposition}

\begin{proof}
This is just a combination of finite LLL and compactness argument (K\"onig's lemma). Indeed, each event from $\mathcal{A}$ is open the the product topology; if the claim
is false, these events cover the entire (compact) product space, so there exists
a finite subset of events that covers the entire space, which contradicts the
finite LLL.
\end{proof}

Our goal is to make this infinite version effective and get a \emph{computable version} of LLL.  First, we need to add some computability conditions. Let us assume that the range of $x_i$ is $\{0,1,\ldots,n_i-1\}$, where $n_i$ is a computable function of $i$. We assume that $x_i$ has a rational-valued probability distribution that is computable given $i$. We assume also that events  $A_i$ are effectively presented, i.e., for a given $i$ one can compute the list of all the variables in $\vbl(A_i)$ and the event itself (i.e., the list of tuples that belong to it).  Finally, we assume that for each variable $x_i$ only finitely many events involve this variable, and the list of those events can be computed given $i$.

\begin{theorem}[Computable infinite version of LLL]
\label{computablelll}
Suppose there is a rational constant $\alpha\in(0,1)$ and a
computable sequence $z_1,z_2,\ldots$ of rational numbers in $(0,1)$ such that
     $$
\Pr [A_i]\le \alpha z_i\cdot\!\!\!\!\!\!\!\!\prod_{j\ne i,\, A_j\in N(A_i)}(1-z_j)
    $$
for all $i$. Then there exists a computable assignment that
avoids all $A_i$.
\end{theorem}

Note that the computability restrictions look quite naturally and that we only need
to make the upper bounds for probability just a bit stronger (adding some constant factor $\alpha<1$). It should not be a problem for typical applications of LLL; usually this stronger bound on $\Pr[A_i]$ can be easily established.

This theorem is the main result of the paper. Lance Fortnow formulated this statement and conjectured that it should follow somehow from the Moser--Tardos effective proof of finite LLL~\cite{moser-tardos}.  It turned out that it is indeed the case (the proof, found by the first author, was published as \texttt{arxiv} preprint~\cite{rumyantsev-positive}). The  proof goes as follows: according to Proposition~\ref{closed-almost-everywhere}, it is enough to construct a computable probability distribution on infinite sequences such that every $A_i$ has zero probability. To construct such a distribution, we need first to extend the class of probabilistic machines, and then construct a machine of this extended type, based on Moser--Tardos algorithm.

\section{Rewriting machines and layerwise computability}

Now we change the computational model and make the output tape of a probabilistic machine (without input) \emph{rewritable}: the machine can change several times the contents of a cell, and only the final value matters. (We may assume that $i$th cell may contain integers in $0\ldots n_i-1$ range, and that initially all cells contain zeros.) The final value is defined if a cell changes its value finitely many times during the computation. Of course, for some values of random bits it may happen that some cell gets infinitely many changes. In this case we say that the output sequence of the machine is undefined.

If the output sequence is defined with probability $1$, we get an almost everywhere defined mapping from Cantor space (sequences of bits produced by coin tossing) into the space of all assignments: the sequence of randomly generated bits is mapped to the output sequence. This mapping defines the output distribution on the assignments (the image of the uniform distribution on random bits). This distribution may be non-com\-put\-able (e.g., for every lower semicomputable real $\Omega\in(0,1)$ it is easy to generate $0000\ldots$ with probability $\Omega$ and $1111\ldots$ with probability $1-\Omega$). However, we get a computable output distribution if we impose some restriction on the machine.

Here it is: \emph{for every output cell $i$ and for every rational $\delta>0$ one can effectively compute integer $N=N(i,\delta)$ such that the probability of the event ``the contents of $i$-th cell changes after step $N$'', taken over all possible random bit sequences, does not exceed~$\delta$}.

\begin{proposition}\label{output-layerwise}
 In this case the limit content of every output cell is well defined with probability $1$, and the output distribution on the sequences of zeros and ones is computable.
\end{proposition}

\begin{proof}
Indeed, to approximate the probability of the event ``output starts with $u$'' for a sequence $u$ of length $k$ with error at most $\delta$, we find $N(i,\delta/k)$ for $k$ first cells (i.e., for $i=0,1,\ldots, k-1$). Then we take $n$ greater than all these values of $N$, and simulate first $n$ steps of the machine for all possible combinations of random bits.
\end{proof}

An almost everywhere defined mappings of Cantor space defined by machines with described properties, are called \emph{layerwise computable mappings}. Initially they appeared in the context of algorithmic randomness~\cite{hoyrup-rojas}. One can show that such a mapping is defined on all Martin-L\"of random sequences. Moreover, it can be computed by a machine with a write-only output tape if the machine additionally gets the value of randomness deficiency for the input sequence. This property can be used (and was originally used) as an equivalent definition of layerwise computable mapping. (The word ``layerwise'' reflects this idea: the mapping is computable on all ``randomness levels''.)

This aspect, however, is not important for us now; all we need is to construct a layerwise computable mapping whose output distribution avoids all the undesirable events $A_i$ with probability~$1$, and then apply Proposition~\ref{closed-almost-everywhere}.

\section{Moser--Tardos algorithm and its properties}

\subsection{The algorithm}

First, we recall the Moser--Tardos probabilistic algorithm that finds a satisfying assignment for the finite case. To make this paper self-contained, we reproduce the arguments from~\cite{moser-tardos} (we need to use them in a slightly modified form).

\medskip\noindent
Moser--Tardos algorithm is very simple and natural:

\begin{itemize}
\item start by assigning random values to all $x_1,x_2,\ldots$ independently (according to their distributions);
\item while some of the events $A_i$ are true, select one of these events and resample all the variables in $\vbl(A_i)$ using fresh independent random values.
\end{itemize}

\noindent
Several remarks:

\begin{itemize}

\item The resampling caused by some $A_i$ does not necessarily makes this event false; it may remain true and can be resampled again later.

\item Resampling for $A_i$ can affect its neighbors; in particular, some of them could be false before and become true after the resampling (then further resampling is needed).

\item The algorithm may terminate (when the current assignment makes all $A_i$ false) or continue indefinitely. The latter can happen even for finite case (though with zero probability, under the conditions of LLL, as we will see). In the infinite case the algorithm does not terminate: even if all $A_i$ are false at some moment, we need an infinite amount of time to check them all.

\item To describe the algorithm fully, one needs to specify some (deterministic or probabilistic) rule that says which of the (true) $A_i$ should be resampled. We assume that some rule is fixed. In the finite case the choice of the rule does not affect the analysis at all; for the infinite case we assume that the \emph{first true $A_i$ is selected} (the one with minimal $i$).

\end{itemize}

\subsection{The analysis: outline}
Let us now reproduce Moser--Tardos analysis for the finite case. They show that the expected number of steps (resampling operations) does not exceed
   $$\sum_i \frac{z_i}{1-z_i}.$$
This implies that the algorithm terminates with probability $1$. How do they get this bound?

\begin{itemize}

\item For each resampling performed during the algorithm, some tree is constructed. Its vertices are labeled by events $A_i$; informally, the tree reflects the dependence between events that were resampled on this and previous steps.

\item Running Moser--Tardos algorithm step by step, we get a sequence of trees (one per each resampling); all the trees in this sequence will be different. (This sequence is a random variable depending on the random choices made when choosing the random values and the events for resampling.)

\item The number of resampling steps is the number of trees appearing in this sequence, therefore the expected number of steps is
$$\sum_{T\in\Tr} \Pr\,[\text{$T$ appears}]$$
where $\Tr$ includes all trees that may appear. (It is important here that a tree cannot appear twice in the same run of the algorithm, see above.)

\item This sum is now estimated by splitting it accordingly to root labels of the trees. Let $\Tr_i$ be a set of trees in $\Tr$ that have root label $A_i$ (such a tree appears when $A_i$ is resampled). Moser and Tardos show that
$$\sum_{T\in\Tr_i} \Pr\,[\text{$T$ appears}]\le \frac{z_i}{1-z_i};$$
this is done by comparing the probabilities in question to some other probabilities (to appear in a process of Galton--Watson type, see below). Then we sum these inequalities for all~$i$.
\end{itemize}

\subsection{The set $\Tr$ and the construction of the trees}

Let us now go into the details.

\begin{itemize}
\item The set $\Tr$ contains finite rooted trees whose vertices are labeled by events $A_k$, with the following restrictions: the sons of a vertex with label $A_k$ should be labeled by different neighbors of $A_k$ (recall that $A_k$ is a neighbor of itself), and these neighbors should be disjoint events.

\item Now we explain how trees are constructed. Assume that the events $A_{i_1},\ldots, A_{i_k},\ldots$ are resampled (in this order) by the algorithm, and we now are at step $k$ and construct a tree $T\in \Tr$ that corresponds to $A_{i_k}$ (=appears at step $k$ of the algorithm). We consider the resampling operations in the reverse order, starting from $A_{i_k}$. Initially the tree consists of a root vertex only; the root is labeled by $A_{i_k}$ (or just by $i_k$; we may identify events with their indices). Then we consider events $A_{i_{k-1}},\ldots, A_{i_1}$. If the next event $A_s$ in this sequence is not a neighbor of any current tree vertex (its label), we skip it. If it is, we look at all these vertices (labeled by neighbors), take one that is most distant from the root (breaking the ties arbitrarily), and attach to it a new son labeled by $A_s$ (by $s$, if we use indices as labels).
\bigskip

\begin{figure}[h]
\begin{center}
\includegraphics{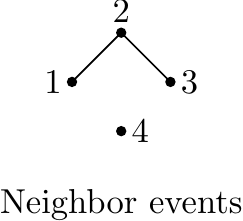}\hspace*{20mm}
\includegraphics{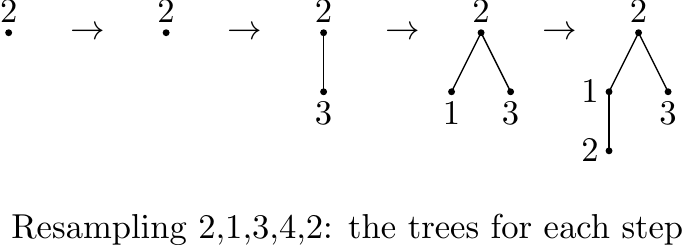}
\end{center}
\end{figure}
An example: consider four events $1,2,3,4$ with neighbor pairs $(1,2)$ and $(2,3)$, and assume that the sequence of resamplings is $2,1,3,4,2$. Then we start with root labeled $2$, then skip $4$ (since it is not a neighbor of $2$), then add $3$ as a son of $2$, then add $1$ as another son of $2$ (note that $1$ and $3$ are not neighbors, so $1$ cannot be attached to $3$), and finally add first $2$ as a son of $1$ or $3$ (the first possibility is shown; note that $2$ is also a neighbor of the root, but we select a vertex on a maximal distance from the root).

\item The resulting tree belongs to $\Tr$. Indeed,

\begin{itemize}
\item each son is a neighbor of its father by construction;
\item different sons of the same father are not neighbors, otherwise the younger brother will become a son of the older one instead of becoming his brother.
\end{itemize}

\item The last argument shows in fact more: any two vertices of the same height (at the same distance from the root) are not neighbors. (This property will be used later.)

\item Two trees that appears at different steps of the algorithm, are different. It is obvious if they have different root labels. If they both have some $A_i$ as the root label, then each tree includes all previous resamplings of $A_i$, so the number of $i$-labels is different for both trees.

\end{itemize}

\subsection{The probability estimate}

Our next goal: for a given $T\in \Tr$ we upperbound the probability of the event ``$T$ appears at some stage of Moser--Tardos algorithm''. This bound is provided by the following lemma.

\begin{lemma}\label{treeprob}
For every tree, the probability of the event ``T appears at some stage of Moser--Tardos algorithm'' does not exceed of the product of $\Pr[A_i]$ for all labels in $T$.
\end{lemma}

\noindent
(If the label $A_i$ appears several times as a label in $T$, its probability appears several times in the product.)

To prove this lemma, it is convenient to specify more explicitly how the random values (used as original values and for resampling) are chosen. Long ago people used printed tables of random numbers. These tables were prepared in advance by a random process; when a fresh random number was needed, it was taken from the table and then crossed out (to prevent its future use). In the same way we assume that for each variable $x_i$ a sequence of independent random values $x_i^0, x_i^1, x_i^2,\ldots$ is prepared in advance. When a new random value is needed, we use the first unused value from this sequence. (So $x_i^0$ is used for the initialization, and $x_i^1,x_i^2,\ldots$ are then used for resampling of events that involve~$x_i$.)

Now the crucial observation: \emph{the tree $T$ that appears at some step, determines uniquely which random values were used during the resampling operations
corresponding to $T$'s vertices}. Indeed, a given variable $x_i$ appears at most once at every height level of the tree (as we have mentioned above), and the height ordering agrees with the temporal ordering of the corresponding resampling operations. No other resampling step (in between the resampling operations included in the tree) could involve $x_i$ due to tree construction (such a resampling operation cannot be skipped).

Therefore, the event ``$T$ appears'' is included in the intersection of the independent events of probability $\Pr[A_i]$ for all labels $A_i$ in $T$. The lemma is proven.

\subsection{Comparison with the Galton--Watson process}

Now we use the assumption of the finite LLL, the upper bound for $\Pr[A_i]$. Using this bound, we need to show that
$$
   \sum_{T\in \Tr_i } \Pr[\text{$T$ appears during the algorithm}]\le \frac{z_i}{1-z_i}.
$$
for every $i$ (We have already noted that this gives the required bound for the expected number of steps.) This is done by the following nice (and mysterious) trick.

Consider the following probabilistic process  (of Galton--Watson type) that produces some tree in $\Tr_i$. The process starts with a tree that has only one vertex, the root, labeled by $i$. Then for each neighbor $A_j$ of $A_i$ (including $A_i$ itself) we decide whether the root will have a son labeled by $j$. (Only one son with a given label can appear.) Namely, son with label $j$ appears with probability $z_j$. Then the same is done for each new vertex $y$: each its neighbor $A_k$ is attached as a son of $y$ with probability $z_k$, and so on. This process may never terminate (then we get an infinite tree), but may also terminate and produce a finite tree.

For example, if event $A_i$ has neighbors $A_i, A_j, A_k$ (and no others), the probability to obtain a tree with only one vertex (the root) in this process is
     $$ (1-z_i)(1-z_j)(1-z_k)$$
(this tree appears when none of three possible sons are attached). And the probability to get a tree with root $i$ and its son $j$ (and no other vertices) is
     $$(1-z_i)z_j(1-z_k)\cdot(1-z_i)(1-z_j),$$
assuming that $A_j$ has no other neighbors except for $A_i$ and $A_j$ itself.

Now the Galton--Watson (GW-) process (for a given $i$) is described, and we claim that for every tree $T\in \Tr_i$ (with root label $i$) we have
    $$
\Pr[\text{$T$ appears during the algorithm}]\le\frac{z_i}{1-z_i}\Pr[\text{$T$ appears in the GW-process}].
    $$
Since only one tree can appear in GW-process, the sum of probabilities in the right hand side (over all $T\in\Tr_i$) does not exceed~$1$, and we get the promised bound.

So it remains to prove our claim. According to the Lemma~\ref{treeprob}, the upper bound for the left-hand side can be obtained by multiplying the bounds for $\Pr[A_i]$ for all labels $A_i$ that appear in the tree.  In our example, if the tree consists of root $i$ alone, we get (using the assumption of LLL)
   \begin{multline*}
   \Pr[\text{$T$ appears during the algorithm}]\le z_i(1-z_j)(1-z_k)=\\
   = \frac{z_i}{1-z_i}[(1-z_i)(1-z_j)(1-z_k)] =
   \frac{z_i}{1-z_i}\Pr[\text{$T$ appears in the GW-process}]
   \end{multline*}
Attaching the son $j$ to the root $i$, we multiply the left hand side by
   $$
\Pr[A_j]\le z_j (1-z_i)=\frac{z_j}{1-z_j}[(1-z_i)(1-z_j)];
   $$
in the right hand side we replace $(1-z_j)$ (probability of not having the son $z_j)$ by $z_j$ (probability of having it) and add new factors $(1-z_i)(1-z_j)$ (since $j$ has no sons in the tree), which gives exactly the same factor. The same happens when we add more vertices.

This finishes the Moser--Tardos analysis. Let us put all the steps together again: knowing that
$$
   \sum_{T\in \Tr_i } \Pr[\text{$T$ appears during the algorithm}]\le \frac{z_i}{1-z_i},
$$
we now sum these inequalities for all $i$ and get
$$
\E[\text{number of steps}]=\sum_{T\in \Tr} \Pr[\text{$T$ appears during the algorithm}]\le \sum_i\frac{z_i}{1-z_i},
$$
as required.
\smallskip

\subsection{Infinite case}

Now we need to analyze the infinite case. We have already described the algorithm; now we need to show that it indeed provides a layerwise computable mapping with required properties. This means that

\begin{itemize}

\item
almost surely every variable $x_i$ is changed only finitely many times and, moreover, the probability of the event ``$x_i$ changes after first $t$ steps of the algorithm'' effectively converges to zero as $t\to\infty$;

\item
the final values of the variables avoid all $A_i$ (make them false).

\end{itemize}

To prove both statements, it is enough to show that
$$
\Pr[\text{$A_i$ is resampled after first $t$ steps}]\to 0
$$
as $t\to \infty$, and the convergence is effective (for every rational $\delta>0$ one can compute some $T$ such that this probability is less than $\delta$ when $t>T$). Indeed, every variable is involved in finitely many events (and we can compute the list of these events); note also that if some $A_i$ is false for the final values of the variables, first of these $A_i$ will be resampled infinitely many times (and this happens with zero probability).

We will show this effective convergence in two steps. First, consider some $k$ and the first $k$ events $A_1,\ldots, A_k$. While some of these events are true, the algorithm will never resample events $A_j$ with $j>k$, so it behaves exactly as the finite Moser--Tardos algorithm applied to $A_1,\ldots,A_k$. So we can use the result proved earlier: if $T_k$ is the first moment when all $A_1,\ldots,A_k$ are false,
	$$
\E [T_k] \le \sum_{i=1}^k \frac{z_i}{1-z_i}
	$$
(we do not need to use the additional factor $\alpha$ yet). This means that $T_k$ is finite almost surely, and the probability of $T_k > t$ converges to $0$ effectively as $t\to\infty$ (due to Tschebyshev inequality).

However, it is not enough for our purposes. At the moment $T_k$ all the events $A_1,\ldots,A_k$ are false, but some of the events $A_j$ with $j>k$ may be true. Resampling such $A_j$, we may change the variables that appear in $A_1,\ldots,A_k$, and this change may make some of these events true again, thus triggering further resampling operations. So we need a more detailed analysis.

Let us look how some $A_i$ among $A_1,\ldots,A_k$ may become true again at some moment $t>T_k$. This can happen only if some neighbor of $A_i$, say, $A_p$, was resampled between $T_k$ and $t$. At that moment $A_p$ was true, so the are two possibilities:
\begin{itemize}
\item $p>k$ (then $A_p$ can be true at the moment $T_k$);
\item $p\le k$, in this case $A_p$ was false at the moment $T_k$ and became true later.
\end{itemize}
In the second case ($A_p$ was false and then became true) some neighbor of $A_p$, say, $A_q$, was resampled before $A_p$ became true. Again, we have two cases $q>k$ and $q\le k$; in the latter case $A_q$ became true because some neighbor $A_r$ of $A_q$ was resampled, etc. We come to the following crucial observation: \emph{if $A_i$ is resampled after the moment $T_k$, there is a chain of neighbors that starts with $A_i$ and ends with some $A_s$ with $s>k$, and all elements of this chain appear in the tree for the $A_i$-resampling}. For fixed $i$ and very large $k$ the length of this chain (and therefore the size of the tree) is big, since all events that are close to $A_i$ are included in $A_1,\ldots,A_k$. And for big trees the additional factor $\alpha$ decreases significantly the probability of their appearance.

Indeed, the inequality
$$
\Pr[\text{$T$ appears during the algorithm}]\le\frac{z_i}{1-z_i}\Pr[\text{$T$ appears in the GW-process}]
    $$
 (for every $T\in\Tr$) now gets an additional factor $\alpha^{\size(T)}$ in the right hand side (since we add factor $\alpha$ for every vertex). We can now proceed as follows. Fix some $A_i$. For an arbitrary natural $m$ take $k$ large enough, so that $A_1,\ldots,A_k$ contain all events that are at distance $m$ or less from $A_i$ in the neigborhood graph. Then the event ``$A_i$ is resampled after time $t$'' is covered by two events:

\begin{itemize}
\item $t<T_k$;
\item $A_i$ was resampled after $T_k$.
\end{itemize}

The probability of the first event ($T_k>t$) effectively converges to $0$ as $t$ increases, so it remains to bound the probability of the second event. As we have seen, this event happens only if a large tree ($m$ or more vertices) with root $i$ appears during the algorithm, so this probability is bounded by
 \begin{multline*}
 \sum_{T\in\Tr_i,\ \size(T)\ge m} \Pr[\text{$T$ appears during the algorithm}]\le\\
\le \sum_{T\in\Tr_i} \alpha^m \Pr[\text{$T$ appears during the GW-process}]\le \alpha^m,
\end{multline*}
so we get what we wanted.

\medskip
This finishes the proof of Theorem~\ref{computablelll}.

\section{Corollaries}

We conclude the paper by showing some special cases where computable LLL can be applied.

A standard illustration for LLL is the following result:
\emph{a CNF where each clause contains $m$ different variables and
has at most $2^{m-2}$ neighbors, is always satisfiable}.
Here neighbors are clauses that have common variables.

Indeed, we let $z_i=2^{-(m-2)}$ for all the events and note that
  $$
2^{-m} \le 2^{-(m-2)} [(1-2^{-(m-2)})^{2^{m-2}}],
  $$
since the expression in square brackets is approximately $1/e>1/2^2$.

This was about finite CNFs; now we may consider \emph{effectively presented} infinite CNFs. This means that we consider CNFs with countably many variables and clauses (numbered by natural numbers); we assume that for given $i$ we
can compute the list of clauses where $i$th variable appears, and for given
$j$ we can compute $j$th clause.

\begin{corollary}
    For every effectively presented infinite CNF where each clause contains $m$ different variables and has at most $2^{m-2}$ neighbors, one can find a computable
    assignment that satisfies it.
\end{corollary}

\begin{proof}
Indeed, the same choice of $z_i$ works, if we choose $\alpha$ close to $1$ (say, $\alpha=0.99$).
\end{proof}

Similar argument can be applied in the case where there are clauses
of different sizes. The condition now is as follows: for every variable there
are at most $2^{\gamma m}$ clauses of size $m$ that involve this variable,
where $\gamma\in(0,1)$ is some constant. Note that now we do not assume
that every variable appears in finitely many clauses, so the notion of
effectively presented infinite CNF should be extended: we assume that for each
$i$ and for each $m$ one can compute the list of clauses of size $m$ that
include $x_i$.

\begin{corollary}\label{variable-cnf}
    For every $\gamma\in(0,1)$ there exists some $M$ such that every effectively presented infinite CNF where each variable appears in at most $2^{\gamma m}$ clauses of size $m$ \textup(for every $m$\textup) and all clauses have size at least $M$, has a computable satisfying assignment.
\end{corollary}

\begin{proof} Let us consider first a special case when each variable appears
only in finitely many clauses. Then we are in the situation covered by
Theorem~\ref{computablelll}, and we need only to choose the values of $z_i$.
These values will depend on the size of the clause: let us choose
  $$
z_i=2^{-\beta k}
  $$
for clauses of size $k$, where $\beta$ is some constant. In fact, any constant
between $\gamma$ and $1$ will work, so we can use, e.g., $\beta=(1+\gamma)/2$.
So we need to check (for clause $A$ of some size~$k$ and for some constant $\alpha<1$) that
  $$
2^{-k} \le \alpha 2^{-\beta k} \prod_{B\in N(A)} (1-2^{-\beta\#B})
  $$
(in fact we can omit the clause $B=A$ in the product, but this does not matter).
Note that for each of $k$ variables in $A$ there are at most $2^{\gamma m}$
clauses of size $m$ that involve it. So together $A$ has at most $k2^{\gamma m}$
neighbors of size $m$. So it is enough to show that
  $$
2^{-k} \le \alpha2^{-\beta k} \prod_{m\ge M} (1-2^{-\beta m})^{k2^{\gamma m}}
  $$
Taking $k$th roots (we replace $\sqrt[k]{\alpha}$ by $\alpha$, but this only makes the requirement stronger) and using that $(1-h)^s \ge 1-hs$, we see that it is enough
to show that
  $$
2^{-1} \le \alpha 2^{-\beta} \left(1-\sum_{m\ge M}2^{\gamma m} 2^{-\beta m}\right)
  $$
Since the series $\sum 2^{(\gamma-\beta)m}$ is converging, this is
guaranteed for large $M$ and for $\alpha$ sufficiently close to $1$.

So we have proven Corollary~\ref{variable-cnf} for the special case when
each variable appear only in finitely many clauses (and we can compute
the list of those clauses).

The general case is easily reducible to this special one. Indeed, fix some $\rho>0$ and delete from each clause $\rho$-fraction of its variables with minimal indices. The CNF becomes only harder to satisfy. But if $\rho$ is small enough, the condition of the theorem (the number of clauses with $m$ variables containing a given variable is bounded by $2^{\gamma m}$) is still true for some $\gamma'\in (\gamma,1)$, because the deletion makes the size of clauses only slightly smaller and decreases the set of clauses containing a given variable. And in this modified CNF each variable appears only in clauses of limited size (it is deleted from all large enough clauses).
\end{proof}

Let us note an important special case. Assume that $F$ is a set of binary strings that contains (for some fixed $\gamma<1$ and for every $m$) at most $2^{\gamma m}$ strings of length $m$. Then one can use LLL to prove the existence of an infinite (or bi-infinite) sequence $\omega$ and a number $M$ such that $\omega$ does not have substrings in $F$ of length greater than $M$. There are several proofs of this statement;  the most direct one is given in~\cite{miller-two-notes}; one may also use LLL or Kolmogorov complexity, see~\cite{rumyantsev-1,rumyantsev-2}.

Joseph Miller noted that his proof (given in~\cite{miller-two-notes}) can be used to show that for a decidable (computable) set $F$ with this property one can find a computable $\omega$ that avoids long substrings in $F$. Konstantin Makarychev extended this argument to bi-infinite strings (personal communication). Now this result becomes a special case of Corollary~\ref{variable-cnf}: places in the sequence correspond to variables, each forbidden string gives a family of clauses (one per each starting position), and there is at most $m2^{\gamma m}$ clauses of size $m$ that involve a given position (and this number is bounded by $2^{\gamma'm}$ for slightly bigger $\gamma'<1$ and large enough $m$).

Moreover, we can do the same for 2-dimensional case: having a decidable
set $F$ of rectangular patterns that contains at most $2^{\gamma m}$ different
patterns of size (=~area) $m$, one can find a number $M$ and a computable 2D configuration (a mapping $\mathbb{Z}^2\to\{0,1\}$) that does not contain patterns from $F$ of size $M$ or more. (It is not clear how to get this result directly, without using Moser--Tardos technique.)

Authors are grateful to Lance Fortnow who suggested to apply Moser--Tardos
technique to get the infinite computable version of LLL, and to the anonymous referees for their comments.

\end{document}